\documentclass[sigconf]{acmart}
\AtBeginDocument{%
  \providecommand\BibTeX{{%
    \normalfont B\kern-0.5em{\scshape i\kern-0.25em b}\kern-0.8em\TeX}}}

\copyrightyear{2021}
\acmYear{2021}
\setcopyright{acmlicensed}\acmConference[SIGIR '21]{Proceedings of the 44th International ACM SIGIR Conference on Research and Development in Information Retrieval}{July 11--15, 2021}{Virtual Event, Canada}
\acmBooktitle{Proceedings of the 44th International ACM SIGIR Conference on Research and Development in Information Retrieval (SIGIR '21), July 11--15, 2021, Virtual Event, Canada}
\acmPrice{15.00}
\acmDOI{10.1145/3404835.3462966}
\acmISBN{978-1-4503-8037-9/21/07}

\usepackage{amsthm}
\newtheorem{definition}{Definition}
\usepackage{color}
\usepackage{multirow}
\usepackage{multicol}
\usepackage{graphicx}
\usepackage{enumitem}
\usepackage{subcaption}
\usepackage[noend]{algpseudocode}
\usepackage[ruled, linesnumbered]{algorithm2e}
\usepackage{pgfplots}\pgfplotsset{compat=1.9}
\usepackage[skip=5pt]{caption}

\begin{document}

% \title{Personalized Counterfactual Fairness in Recommendation}
% \title{Towards Personalized Fairness}
% \title{Towards Personalized Fairness based on Causal Notion}
\title{Personalized Counterfactual Fairness in Recommendation}
% \author{Yunqi Li}
% \affiliation{
%   \institution{Rutgers University}
%   \city{New Brunswick}
%   \state{NJ}
% \email{yunqi.li@rutgers.edu}

% \author{Hanxiong Chen}
% \affiliation{
%   \institution{Rutgers University}
%   \city{New Brunswick}
%   \state{NJ}
% \email{hanxiong.chen@rutgers.edu}

% \author{Shuyuan Xu}
% \affiliation{
%   \institution{Rutgers University}
%   \city{New Brunswick}
%   \state{NJ}
% \email{shuyuan.xu@rutgers.edu}

% \author{Yingqiang Ge}
% \affiliation{
%   \institution{Rutgers University}
%   \city{New Brunswick}
%   \state{NJ}
% \email{yingqiang.ge@rutgers.edu}

% \author{Yongfeng Zhang}
% \affiliation{
%   \institution{Rutgers University}
%   \city{New Brunswick}
%   \state{NJ}
% \email{yongfeng.zhang@rutgers.edu}

\author{Yunqi Li, Hanxiong Chen, Shuyuan Xu, Yingqiang Ge, Yongfeng Zhang}
\affiliation{
  \institution{Computer Science, Rutgers University, New Brunswick, NJ 08854, USA}
  \country{}
}
\email{{yunqi.li, hanxiong.chen, shuyuan.xu, yingqiang.ge, yongfeng.zhang}@rutgers.edu}

\begin{abstract}
Recommender systems are gaining increasing and critical impacts on human and society since a growing number of users use them for information seeking and decision making. Therefore, it is crucial to address the potential unfairness problems in recommendations.

Just like users have personalized preferences on items, users' demands for fairness are also personalized in many scenarios. Therefore, it is important to provide \textit{personalized} fair recommendations for users to satisfy their \textit{personalized} fairness demands. Besides, previous works on fair recommendation mainly focus on association-based fairness. However, it is important to advance from associative fairness notions to causal fairness notions for assessing fairness more properly in recommender systems. Based on the above considerations, this paper focuses on achieving personalized counterfactual fairness for users in recommender systems. To this end, we introduce a framework for achieving counterfactually fair recommendations through adversary learning by generating feature-independent user embeddings for recommendation. The framework allows recommender systems to achieve personalized fairness for users while also covering non-personalized situations. Experiments on two real-world datasets with shallow and deep recommendation algorithms show that our method can generate fairer recommendations for users with a desirable recommendation performance.

\end{abstract}

\begin{CCSXML}
<ccs2012>
<concept>
<concept_id>10010147.10010257</concept_id>
<concept_desc>Computing methodologies~Machine learning</concept_desc>
<concept_significance>500</concept_significance>
</concept>
<concept>
<concept_id>10002951.10003317.10003347.10003350</concept_id>
<concept_desc>Information systems~Recommender systems</concept_desc>
<concept_significance>500</concept_significance>
</concept>
</ccs2012>
\end{CCSXML}

\ccsdesc[500]{Computing methodologies~Machine learning}
\ccsdesc[500]{Information systems~Recommender systems}

\keywords{Personalized Fairness; Counterfactual Fairness; Recommender System; Adversary Learning}


\maketitle

\section{Introduction}

% As one of the most extensive applications of machine learning in industry, recommender systems now play an increasingly significant role with lots of services using them to match products and users. Therefore, it is important and necessary to identify and address the potential unfair issue in recommender systems which may harm the interests of users. 
% \begin{figure}
%     \centering
%     \includegraphics[scale=0.8]{figure/pic1.png}
%     \caption{Example of users' personalized fairness demands.}
%     \label{fig:pic1}
% \vspace{-20pt}
% \end{figure}

Recently, there has been growing attention on fairness considerations in recommendation models 
\cite{xiao2017fairness,yao2017beyond, burke2017multisided, leonhardt2018user, beutel2019fairness,li2021user,ge2021towards,fu2020fairness}. The fairness problem in recommender systems---which are known as multi-stakeholder platforms---should be considered from different perspectives, including user-side, item-side or seller-side \cite{burke2017multisided}. Compared with the many fairness research on item- or seller-side \cite{adomavicius2011improving,kamishima2014correcting,abdollahpouri2017controlling,abdollahpouri2019managing,singh2018fairness}, fairness issues on the user-side has been less studied in recommender systems. One challenge for user-side fairness research---compared to the item-side---is that different users' fairness demands can be different due to their personalized preferences. For example, as shown on Figure.\ref{fig:pic1}, some users may be sensitive to the gender and do not want their recommendations to be influenced by this feature, while others may care more about the age feature and are less concerned about gender. As a result, it is important to explore \textit{personalized fairness} in recommendation scenarios.
% Making a decision making system personalized is a key to improve user experience and thus increase gains. 
However, many existing works consider fairness on the same set of sensitive features for all users, and personalized fairness demands are largely neglected. To better assess the unfairness issues in recommendation, it is important to enhance fairness from the personalized view.
% we should consider causal-based fairness notions on users side in recommender systems instead of just association-based fairness concepts; on the other hand, as users may have personalized demands for fairness in many cases, 
% it would be more reasonable to consider 
% user fairness in the individual level. Based on the above considerations, we focus on achieving counterfactual fairness for users in this paper.

\begin{figure}[t]
    \centering
    \includegraphics*[scale=1.0]{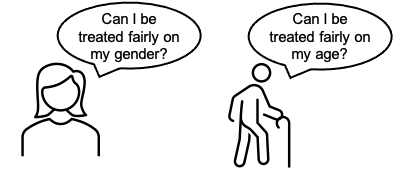}
    \vspace{-10pt}
    \caption{Example of users' personalized fairness demands.}
    \label{fig:pic1}
\vspace{-10pt}
\end{figure}

Besides, existing work about achieving fairness in recommendations are mainly based on association-based fairness notions, which primarily focus on discovering the discrepancy of statistical metrics between individuals or sub-populations. For example, equalized odds \cite{hardt2016equality}, which is one of the most basic criteria for fairness, requires that the false positive rate and true positive rate should be equal for protected group and advantaged group. However, recent research show that fairness cannot be well assessed only based on association notions \cite{khademi2019fairness,kusner2017counterfactual,zhang2018equality,zhang2018fairness}.
A classic example is the Simpson’s paradox \cite{pearl2009causality}, where the statistical conclusions drawn from the sub-populations and the whole population can be different. In the context of fairness modeling, association-based notions cannot reason about the causal relations between the protected features and the model outcomes. Different from association-based notions, causal-based notions leverage prior knowledge about the world structure in the form of causal models, and thus can help to understand the propagation of variable changes in the system. Therefore, causal-based fairness notions are more and more important to addressing discrimination in machine learning models \cite{khademi2019fairness,kusner2017counterfactual,zhang2018equality,zhang2018fairness}. In this paper, except for the above mentioned personalized view, we also consider fairness in recommendation from a causal view.

To realize personalized \textit{and} causal fairness for recommendation, we pursue personalized counterfactual fairness in this paper. Counterfactual fairness is an individual-level causal-based fairness notion \cite{kusner2017counterfactual}, which considers a hypothetical world beyond the real world. To enable fairness, it requires the probability distributions of the model outcomes to be the same in the factual and the counterfactual world for each individual. This individual-level view makes counterfactual fairness a nice fit for both personalized fairness demands and causal fairness notions. In this paper, we expect a recommender system to be counterfactually fair if the recommendation results for a user are unchanged in the counterfactual world where the user's features remain the same except for certain sensitive features specified by the user. This is to grant users with the right to tell us which features---such as gender, race, age, etc.---that they care about and that they want their recommendations to be irrelevant to.
Technically, we introduce a framework for generating recommendations that are independent from the sensitive features so as to meet the counterfactual fairness requirements. We first analyze how we can guarantee the independence between sensitive features and recommendation outcomes. Specifically, we control the dependence between sensitive features and the user embeddings based on the causal graph of the general recommendation pipeline. To generate feature-independent user embeddings, we introduce an adversary learning approach to remove the information of the sensitive features from the user embeddings, while keeping the embeddings informative enough for the recommendation task. To achieve personalized fairness, we allow each user to select a set of sensitive features that they care about. Finally, we provide two methods---the Separate Method (SM) and the Combination Method (CM)---for generating personalized fair recommendations conditioned on the user's sensitive features.
% based on filter functions.
% for users based on different scenarios. 
% In summary, we aim at generating fair recommendations for users in this paper. Different from the previous works about achieving fairness in recommender systems which condsider fairness from merely the association perspective,  we adopt the causal-based fairness notion in this paper. Specifically, we achieve counterfactual fairness in recommender systems to implement individual level fairness for users. Considering different individual may have personalized demands for fairness, we introduce a framework which allows the recommendation models to generate personalized fairness for users, while also can cover the non-personalized scenario where the user demands for fairness are the same. We first introduce the definition of counterfactual fairness in recommendation, and explain how we can achieve counterfactual fairness through generating feature independent user embeddings. For realizing this goal, we provide a method to remove the information of certain sensitive features from user embeddings through adversary learning. 
Our experiments on two real-world datasets with different types of shallow or deep recommendation algorithms show that our method is able to enhance personalized counterfactual fairness in recommendation with a desirable recommendation performance. 

The key contributions of this paper are as follows:

\begin{itemize}

    % in recommender systems, and propose three methods to distinguish users into advantaged and disadvantaged groups with observable information.
    % % % \item Different from most of the previous researches which solve the item popularity bias in recommendations, we consider unfairness recommendation from the perspective of user visibility. We conduct a data-driven observational analysis to show the performance of four traditional recommendation algorithms on three Amazon datasets to argue that such algorithms produce unfair recommendations between different group of users, and thus reduces the overall recommendation performance.
    % \item We show that recommender systems could suffer from generating biased results based on our grouping strategy. 
    \item We consider unfairness issues in recommendation from a personalized perspective to achieve personalized fairness.
     \item We consider unfairness issues in recommendation from a causal perspective to achieve counterfactual fairness. 
    \item We introduce a framework for achieving personalized counterfactual fairness in recommendation based on adversarial learning.
    % \item Under the definition of counterfactual fairness, we provide a framework for allowing the models to generate personalized fair recommendations for users. 
    \item We conduct experiments on two real-world datasets with both shallow and deep models to show the effectiveness of our framework on enhancing fairness in recommendation.
    % that our method can generate counterfactual fair recommendations with an acceptable recommendation quality loss.
\end{itemize}

In the following, we review related work in Section \ref{sec:related}. Before we introduce the method, we provide some preliminaries and notations in Section \ref{sec:notations}. In Section \ref{sec:framework}, we introduce the details of our framework. Experimental settings and results are provided in Section \ref{sec:experiments}. Finally, we conclude this work in Section \ref{sec:conclusions}.

\section{Related Work}
\label{sec:related}

\subsection{Association-based Fairness}
Fairness is becoming more and more important in machine learning \cite{kamishima2012fairness,pedreshi2008discrimination, singh2018fairness}. Overall, there are two basic frameworks for algorithmic fairness, i.e., group fairness and individual fairness. Group fairness requires that the protected group and advantaged group should be treated similarly \cite{pedreschi2009measuring},
while individual fairness requires that similar individuals are treated similarly. Individual fairness is relatively more difficult to precisely define due to the lack of agreement on similarity metrics for individuals in different tasks \cite{dwork2012fairness}. 

The first endeavor to achieve fairness in the community is to develop association-based (or correlation-based) notions for fairness, which aims to find the discrepancy of statistical metrics between individuals or sub-populations. More specifically,
% supervised learning which considers to achieve group fairness usually implies constraints such as equalized odds and demographic parity. Equalized odds requires that the false positive rate and true positive rate should be equal for protected group and advantaged group, and thus represents the equal opportunity principle \cite{hardt2016equality,zafar2017fairness}. Demographic parity, also known as independence or statistical parity \cite{calders2009building}, defines the constraint to require that decisions should be similar around a sensitive feature such as race or gender.
early works about fairness are mostly on classification tasks, which design algorithms that are compatible with fairness constraints \cite{zemel2013learning,woodworth2017learning}. For binary classification, fairness metrics can be expressed by rate constraints, which regularize the classifier's positive or negative rates over different protected groups \cite{narasimhan2018learning,cotter2019two}. For example, demographic parity requires that the classifier's positive rate should be the same across all groups. To achieve fairness, the training objective is usually optimized together with such constraints over fairness metrics 
% and will be solved using constrained optimization algorithms or relaxations methods
\cite{goh2016satisfying,agarwal2018reductions}. 

Some recent works have also considered the fairness of ranking tasks. Some works directly learn a ranking model from scratch \cite{singh2019policy,zehlike2020reducing,narasimhan2020pairwise,xiao2017fairness}, while others consider re-ranking or post-processing algorithms for fair ranking \cite{celis2017ranking, biega2018equity,li2021user}.
% The fairness in ranking can also be considered from group and individual level, while most existing works in ranking measure unfairness at the level of subject groups. 
The fairness metrics for ranking tasks are usually defined over the exposure of items that belong to different protected groups. As summarized in \cite{narasimhan2020pairwise}, such metrics include the unsupervised criteria and the supervised criteria. Unsupervised criteria posit that the average exposure at the top of the ranking list is equal for different groups \cite{celis2017ranking, singh2018fairness,zehlike2020reducing}, while the supervised criteria require the average exposure of item groups to be proportional to their average relevance to the query \cite{biega2018equity,singh2019policy}.
% What's more, to capture unfairness at the level of individual subjects, \citeauthor{biega2018equity}\cite{biega2018equity} claims that no single ranking can achieve individual attention fairness, and achieves amortized fairness by making the attention accumulated across a series of rankings to be proportional to accumulated relevance.

Recommendation algorithm can usually be considered as a type of ranking algorithms. However, it is also special in that personalization is a very fundamental consideration for recommendation. As a result, different from previous fairness ranking algorithms which usually consider fairness from the item-side, we also need to consider fairness on the user-side in recommendation, as well as users' personalized fairness demands.

\subsection{Causal-based Fairness}

Recently, researches have noticed that fairness cannot be well assessed merely based on correlation or association \cite{khademi2019fairness,kusner2017counterfactual,zhang2018equality,zhang2018fairness}, since they cannot reason about the causal relations between input and output. However, real discrimination may result from a causal relation between the model decisions (e.g. hiring and admission) and the sensitive features (e.g. gender and race). Therefore, 
% As research community to realize the necessity to investigate the causal relationship rather than associated relationship between the sensitive features and the decisions,
causal-based fairness notions are proposed. Causal-based fairness notions are mostly defined on intervention or counterfactual. Intervention can be achieved through random experiments, while counterfactual considers a hypothetical world beyond the real world. We introduce some important causal notions as follows.

% Intervention and counterfactual are non-observable quantities, and can not always be uniquely calculated from the observed data, which seriously hinders the practicability of causal-based notions in real scenes. 

Total effect (TE) \cite{pearl2009causality} measures the effect of changing the sensitive feature along all the causal paths to the outcome. Treatment on the treated (ETT) \cite{pearl2009causality}, which is the most basic fairness notion under counterfactuals, measures the difference between the real world and the counterfactual world where the sensitive feature changes for the individual. Both TE and ETT seek the equality of outcomes between protected and unprotected groups. Disparate treatment \cite{barocas2016big} is another framework which aims at ensuring the equality of treatment by prohibiting the use of sensitive features when making decisions, including direct effect, indirect effect and path-specific effect. Direct discrimination is measured by the causal effect along the causal path from the sensitive feature to the final decision \cite{pearl2013direct}; indirect discrimination is assessed by the causal effect along the causal path through proxy features \cite{pearl2013direct}; while path-specific effect \cite{pearl2009causality} characterizes the causal effect over specific paths. Based on this, various causal-based notions have been put forward. Examples include: No unresolved discrimination \cite{kilbertus2017avoiding}, which measures the indirect causal effects from sensitive features to outcomes and requires that there is no directed path from sensitive features to outcomes except via a resolving variable;
% No proxy discrimination \cite{kilbertus2017avoiding} is satisfied if there doesn't exist any path which is blocked by a proxy variable from sensitive features to outcomes;
% FACE and FACT are two group-based fairness notions proposed by \cite{khademi2019fairness}. FACE cares about the average causal eﬀect of sensitive feature on outcomes at a population level while FACT considers the same eﬀect on the sub-population level;
Equality of effort \cite{huan2020fairness}, which measures how much efforts are needed from the protected group or individual to reach a certain level of outcome to identify discrimination; PC-Fairness \cite{wu2019pc}, which can cover lots of causal-based fairness notions by tuning its parameters. A more comprehensive list of causal-based fairness notions are provided in \cite{makhlouf2020survey}.

This paper aims at achieving counterfactual fairness in recommendation.
Counterfactual fairness \cite{kusner2017counterfactual} is a fine-grained variant of ETT conditioned on all features. It requires the probability distribution of the outcome to be the same in the factual and counterfactual worlds for every individual. We will introduce the definition of counterfactually fair recommendation in detail in the preliminaries. 

\subsection{Fair Recommendation}
Different from fair classification and ranking, the concept of fairness in recommendation can be more complex as it extends to multiple stakeholders \cite{burke2017multisided}.
% which means that the unfair issue can be considered not only from the item or the provider side, but also can be considered from the user side.As researchers are realizing the importance of fairness in recommendations, 
Recent works on fairness in recommendations have very different views.
\citeauthor{xiao2017fairness} \cite{xiao2017fairness} introduced an optimization framework for fairness-aware group recommendation based on Pareto Efficiency. \citeauthor{yao2017beyond} \cite{yao2017beyond} explored fairness in collaborative filtering recommender systems, which proposed four metrics to assess different types of fairness by adding fairness constraints to the learning objective. \citeauthor{burke2017multisided}
\cite{burke2017multisided} and  \citeauthor{abdollahpouri2019multi}\cite{abdollahpouri2019multi} categorized different types of multi-stakeholder platforms and introduced several desired group fairness properties. \citeauthor{leonhardt2018user} \cite{leonhardt2018user} identified the unfairness issue for users in post-processing algorithms to improve the diversity in recommendation. \citeauthor{mehrotra2018towards}
\cite{mehrotra2018towards} proposed a heuristic strategy to jointly optimize fairness and performance in two-sided marketplace platforms. \citeauthor{beutel2019fairness} \cite{beutel2019fairness} considered fairness in recommendation under a pairwise comparative ranking framework, and offered a regularizer to improve fairness when training recommendation models. \citeauthor{patro2020fairrec} \cite{patro2020fairrec} explored individual fairness for both producers and customers for long-term sustainability of two-sided platforms. 
\citeauthor{fu2020fairness} \cite{fu2020fairness} impaired the group unfairness problem in the context of explainable recommendation \cite{zhang2014explicit,zhang2020explainable} over knowledge graphs; \citeauthor{li2021user} \cite{li2021user} considered user-oriented fairness in recommendation by requiring the active and inactive user groups be treated similarly; \citeauthor{ge2021towards} \cite{ge2021towards} proposed a reinforcement learning framework to deal with the changing group labels of items to achieve long-term fairness in recommendation. To the best of our knowledge, our work is the first to consider personalized and causal-based fairness in recommender systems.
% Besides, we introduce a framework to allow users to select different sensitive features they care about and generate personalized fairness in recommendation.

% \citeauthor{fu2020fairness} \cite{fu2020fairness} propose a fairness constrained approach to mitigate the unfairness problem in the context of explainable recommendation over knowledge graphs. They find that performance bias exists  between different user groups, and claim that such bias comes from the different distribution of path diversity. Here, we show that such recommendation performance bias also exists in general recommender systems. There are more researches concerning the popularity bias problem in recommendations, i.e., the frequently rated items will get more exposure than those less popular ones. Such researches mainly solve this problem by increasing the number of recommended unpopular items (long tail items) or otherwise the overall catalog coverage \cite{adomavicius2011improving,kamishima2014correcting,abdollahpouri2017controlling,abdollahpouri2019managing}.
% \citeauthor{abdollahpouri2019unfairness} \cite{abdollahpouri2019unfairness} see problem from the users’ perspective with finding how popularity bias causes the recommendations to deviate from what the user expects to get from the recommender system. In this paper,  we concern the unfair issue caused by the bias in user side.

\section{Preliminaries And Notations}
\label{sec:notations}

In this section, we introduce the preliminaries and notations used in this paper. Capital letters such as $Z$ denote variables, lowercase letters such as $z$ denote specific values of the variables. Bold capital letters such as $\textbf{Z}$ denote a set of variables, while bold lowercase letters such as $\textbf{z}$ denote a set of values. In the following, we first show the notations used in the recommendation task, and then we introduce the preliminaries about counterfactual fairness. 
\vspace{-10pt}

\subsection{Recommendation Task}

In recommendation task \cite{chen2021neural}, we have a user set $\mathbb{U}=\left\{u_{1}, u_{2}, \cdots, u_{n}\right\}$ and an item set $\mathbb{V}=\left\{v_{1}, v_{2}, \cdots, v_{m}\right\}$, where $n$ is the number of users and $m$ is the number of items. The user-item interaction histories are usually represented as a 0-1 matrix $H =\left[h_{i j}\right]_{n \times m}$, where each entry $h_{ij}=1$ if user $u_i$ has interacted with item $v_j$, otherwise $h_{ij}=0$. The key task for recommendations is to predict the preference scores of users over items, so that the model can recommend each user $u_{i}$ a top-$N$ recommendation list $\left\{v_{1}, v_{2}, \cdots, v_{N} | u_{i}\right\}$ according to the predicted scores. To learn the preference scores, modern recommender models are usually trained to learn the user and item representations based on the user-item interactions, and then take the representations as input to a learned or designed scoring functions to make recommendations. 
We use $\textbf{r}_u$ and $\textbf{r}_v$ to represent the learned vector embeddings for user $u$ and item $v$, and use $S_{uv}$ to denote the predicted preference score for a $(u,v)$ pair. In addition to the interaction records, users have their own features, such as gender, race, age, etc. In particular, we use $\textbf{Z}$ to represent the sensitive features, and use $\textbf{X}$ to denote all the remaining features which are not causally dependent on $\textbf{Z}$, i.e., the insensitive features. Without loss of generality, we suppose each user have $K$ categorical sensitive features $\{Z_1,Z_2,...,Z_K\}$.

% $\{Z_1=z_{u}^{1},Z_2=z_{u}^{2},...,Z_K=z_{u}^{K}\}$.

% $a_{u}^{k} \in A_k, k=1,2,...,K$

\vspace{-15pt}
\subsection{Counterfactual}

Before introducing the definition of counterfactual fairness, we first briefly introduce the concept of counterfactual. To understand counterfactual, let us consider an example first \cite{pearl2016causal}. When Alice was driving home, she came to a fork in the road and had to make a choice: to take the street 1 ($X$ = 1) or to take the street 0 ($X$ = 0). Alice took the street~0 and it took her 2 hours to arrive home, and then she may ask ``how long would it take if I had taken the street 1 instead?'' Such a ``what if'' statement in which the ``if'' portion is unreal or unrealized, is known as a counterfactual \cite{pearl2016causal}. The ``if'' portion of a counterfactual is called the antecedent. We use counterfactual to compare two outcomes under the exact same condition, differing only in the antecedent. To solve the above counterfactual, we denote the driving time of the street 1 by $Y_{X=1}$ or $Y_1$, and the driving time of the street 0 by $Y_{X=0}$ or $Y_0$, then the quantity we want to estimate is $E(Y_{X=1}|X=0,Y=Y_{0}=2)$. The counterfactual can be solved based on structural causal model \cite{pearl2009causality}. As the assumption of causal models, the state of $Y$ will be fully determined by the background variables $U$ and the structural equations $F$. Specifically, given $U=u$ (which can be derived from the evidence of $X=0$ and $Y=2$), and an intervention on $X$ as $do(X=1)$,
% the equations for $X$ are replaced with $X=1$, 
we can derive the solution of the counterfactual. To make the notations clear, we use the expression $P(Y_{X=x^{\prime}} =y^\prime | X=x, Y=y)=P(y^\prime_{x^{\prime}} | x, y)$, which involves two worlds: the observed world where $X=x$ and $Y=y$ and the counterfactual world where $X=x^{\prime}$ and $Y=y^\prime$. The expression reads ``the probability of $Y=y^\prime$ had $X$ been $x^{\prime}$ given that we observed $Y=y$ and $X=x$''.

%  We denote it by $Y_{Z \leftarrow z}(u)$, and sometimes as $Y_{z}$ if the context of the notation is clear.The obtained distribution $P(Y=y \mid do(X=x))$ can be considered as a counterfactual distribution since the intervention forces Z to take a value diﬀerent from the one it would take in the actual world.  $P(Y=y \mid d o(Z=z))=P\left(Y_{Z=z}=y\right)=P\left(Y_{z}=\right.$$y)=P\left(y_{z}\right)$ is used to define the causal effect of $z$ on $Y .$ The term counterfactual quantity is used for expressions that involve explicitly multiple worlds. In Figure $2(\mathrm{~b}),$

\subsection{Counterfactual Fairness}

Counterfactual fairness is an individual-level causal-based fairness notion \cite{kusner2017counterfactual}. It requires that for any possible individual, the predicted result of the learning system should be the same in the counterfactual world as in the real world. The counterfactual world here is the world in which we only make an intervention on user's sensitive features, while all other features of the user that are not dependent on the sensitive features are kept unchanged. The counterfactual fairness is an individual-level notion because it is conditioned on all the unchanged variables. Here is an example for counterfactual fairness: suppose we are designing a decision making system that does not discriminate against gender when deciding students' admission to college.
Counterfactual fairness requires that the admission result to a student will not be changed if his or her gender were reversed while all other features that are not dependent on gender remain the same, such as grades and recommendation letters. Such a fairness requirement is usually more reasonable than forcefully requiring the same admission rate for all genders in association-based notions, since students of different genders may have different preferences for college majors.

In this paper, we consider counterfactual fairness in recommendation scenario. We give the definition of counterfactually fair recommendation as follows.

\begin{definition}[Counterfactually fair recommendation]\label{counterfactual fairness}
A recommender model is counterfactually fair if for any possible user $u$ with features $\textbf{X}=\textbf{x}$ and $\textbf{Z}=\textbf{z}$:
% $$
% P\left(L_{\textbf{A} \leftarrow \textbf{a}}=l \mid \textbf{X}=\textbf{x}, \textbf{A}=\textbf{a}\right)=P\left(L_{\textbf{A} \leftarrow \textbf{a}^{\prime}}=l \mid \textbf{X}=\textbf{x}, \textbf{A}=\textbf{a}\right)
% $$
$$
P\left(L_{\textbf{z}} \mid \textbf{X}=\textbf{x}, \textbf{Z}=\textbf{z}\right)=P\left(L_{\textbf{z}^{\prime}} \mid \textbf{X}=\textbf{x}
, \textbf{Z}=\textbf{z}\right)
$$
for all $L$ and for any value $\textbf{z}^{\prime}$ attainable by $\textbf{Z}$, where $L$ denotes the Top-N recommendation list for user $u$.
\end{definition}

Here $\textbf{Z}$ are the user's sensitive features and $\textbf{X}$ are the features that are not causally dependent on \textbf{Z}. This definition requires that for any possible user, sensitive features \textbf{Z} should not be a cause of the recommendation results. Specifically, for a given user $u$, the distribution of the generated  recommendation results $L$ for $u$ should be the same if we only change $\textbf{Z}$ from $\textbf{z}$ to $\textbf{z}^{\prime}$, while holding the remaining features $\textbf{X}$ unchanged.

\section{Fair Recommendation Framework}
\label{sec:framework}

In this section, we introduce the framework for generating counterfactually fair recommendations.

\subsection{Problem Formulation}

As discussed above, we aim at achieving counterfactual fairness in recommendations.
% it is important and necessary for recommender systems to generate fair recommendations for users to improve user satisfactions. To assess fairness problem properly,  we need to consider fairness from a causal view, since association-based fairness notions can not tell the causal relation between protected features and predicted outcomes. What's more, considering that the demands of users for fairness may be personalized in many cases, we consider to address fairness in recommendations in individual level, i.e., the counterfactual fairness in recommendations. 
% Although personalized recommendation is beneficial to improve user satisfaction, users may not want to be treated differently in certain personal features in some cases. For example, in job recommendation, a female user may also want to be recommended some jobs which need physical strength; in music  recommendation, an Asian user may also want be recommended more Europe and America musics rather than just from Asia. Therefore, we need our recommender models to have the ability to recommend personalized items to users, while also treat user fairly on the sensitive features they care about. What's more, different users may care about different sensitive features. We need our model to allow personalized fairness. 
From Definition~\ref{counterfactual fairness}, we can see that counterfactual fairness requires that the generated recommendation results $L$ are independent from the user sensitive features $\textbf{Z}$. As stated in \cite{kusner2017counterfactual}, which considers counterfactual fairness in classification tasks, the most straightforward way to guarantee the independence between predicted outcomes and sensitive features is just avoiding from using sensitive features (and the features causally depend on the sensitive features) as input. However, this is not the case in recommendation scenarios. 
Most of the Collaborative Filtering (CF) \cite{goldberg1992using,ekstrand2011collaborative} or Collaborative Reasoning (CR) \cite{chen2021neural,shi2020neural} recommender systems are directly trained from user-item interaction history, and content-based recommendation models \cite{adomavicius2011context,lops2011content,zhang2017joint} and hybrid models \cite{burke2002hybrid} may use user profiles as input or use additional information to help train the model. However, no matter if the model directly uses user 
features as input or not, the model may generate unfair recommendations on some user features. The reason is that by collaborative learning (CF or CR) in the training data, the model may capture the relevance between user features and user behaviours that are inherently encoded into the training data, since user features may have causal impacts on user behaviors and preferences.
% and the training data may have been polluted due to feedback loops and echo chambers created by the recommendation algorithm in place \cite{jiang2019degenerate,ge2020understanding}.
% For example, in movie recommendation, female users may prefer to click on love movies, while male users may prefer action movies. In music recommendation, Black users may prefer hip-pop musics more than Asian users.
% Therefore, user interaction histories will contain the information of user features, and such information can be captured by learning models.
%  Specifically, recommendation models need to learn user representations from user behaviours. Therefore, the learned user embeddings may contain the information of user features, which will impact the final recommendation results. 
As a result, we need to design methods to achieve counterfactually fair recommendations as it cannot be realized in trivial way.

\begin{figure}[t]
    \centering
    \includegraphics[scale=0.5]{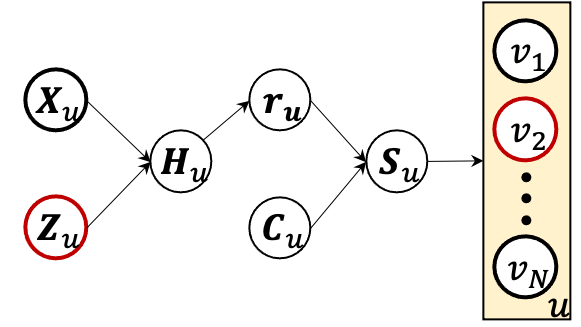}
    \caption{Causal relations for general recommendation models. For a given user $u$, $\textbf{X}_u$ and $\textbf{Z}_u$ are insensitive and sensitive features of $u$, respectively. $\textbf{H}_u$ is the user interaction history. $\textbf{r}_u$ is the user embedding. $\textbf{C}_u$ is the candidate item set for $u$. $\textbf{S}_u$ are the predicted scores over the candidate items. The red circled nodes are used to emphasize the impact of the sensitive features on the final recommendation list.}
    \label{fig2}
    \vspace{-5pt}
\end{figure}

To guarantee that recommendation results are independent from user sensitive features, we only need to require that given a user $u$, for any item $v\in \mathcal{V}$, the predicted score $S_{uv}$ for the user-item pair $(u,v)$ is independent from the user sensitive features $\textbf{Z}$. 
As shown in Figure~\ref{fig2}, which represents the causal relations for general recommendation models, for a given user $u$, the scoring function $\textbf{S}_u$ usually takes user embedding $\textbf{r}_u$ and candidate item embeddings $\textbf{C}_u$ as input to generate the recommendation list. However, the user embedding $\textbf{r}_u$, which is learned from user histories $\textbf{H}_u$, may depend on the user features $\textbf{X}_u$ and $\textbf{Z}_u$ since the features causally impact user behaviours. Therefore, as shown by the causal path from sensitive feature $\textbf{Z}_u$ to the final recommendation result, we only need to ensure the independence between user embedding $\textbf{r}_u$ and the sensitive feature $\textbf{Z}_u$ to meet the counterfactual fairness requirement, i.e., for all $ u \in \mathcal{U}$, we need to guarantee $\textbf{r}_u \perp \textbf{Z}_{u}$.

Besides, to meet users' personalized demands on fairness, we allow each user to select a set of sensitive features that they care about, and we generate fair recommendations in terms of these features. Suppose user $u$ selected a set of sensitive features $\textbf{Q}_u \subseteq\{1, \ldots, K\},$ then we need to guarantee that $\textbf{r}_{u} \perp z_{u}^k$, for all $k\in\textbf{Q}_u$.

% $$
% I\left(\mathbf{r}_{u}, a_{u}^{k}\right)=0, k \in Q, \forall u \in \mathcal{U}
% $$

% $S_{uv}=f(r_u,r_v)$, where $f$ is the scoring function, and $r_u$ and $r_v$ are embeddings of user and item, respectively. Here we reasonably assume that $S_{uv}$ depends on users' sensitive features $\textbf{A}$ only through user embedding $r_u$, since the sensitive feature of an individual should not have an impact on item embedding which represent item characteristic information. Therefore, we only need to enforce the user embedding independent of sensitive features to guarantee that the recommendations results are independent of sensitive features.
% $$
% \mathbf{r}_{u} \perp a_{u}, \quad \forall u \in \mathcal{V}
% $$

% What' more, we know the fact that mutual information between user embedding and the sensitive feature being zero is the “sufficient and necessary condition” for independence, we instead require that the mutual information $I\left(\mathbf{r}_{u}, a_{u}\right)= 0$.

\subsection{The Model}
In this section, we introduce the model to generate feature independent user embeddings through adversary learning. 
% It is a common method about leveraging adversary learning to mitigate discrimination through removing information about sensitive features from representations \cite{xie2017controllable, beutel2017data, elazar2018adversarial, wang2019balanced,bose2019compositional,arduini2020adversarial,du2020fairness}. 
The main idea is to train a predictor and an adversarial classifier simultaneously, where the predictor aims at learning informative representations for the recommendation task, while the adversarial classifier aims at minimizing the predictor’s ability to predict the protected features from the learned representation, and thus the information about sensitive features are removed from the representations to mitigate discrimination \cite{xie2017controllable, beutel2017data, elazar2018adversarial, wang2019balanced,bose2019compositional,arduini2020adversarial,du2020fairness}. Following the adversary learning setup \cite{goodfellow2014generative,bose2019compositional,arduini2020adversarial}, we develop an adversary network that consists two modules: a filter module that aims at filtering out the information about sensitive features from user embeddings, and a discriminator module that aims to predict the sensitive features from the learned user embeddings. Figure~\ref{fig3} shows the architecture of our method.

\begin{figure}[t]
    \centering
    \includegraphics[scale=0.7]{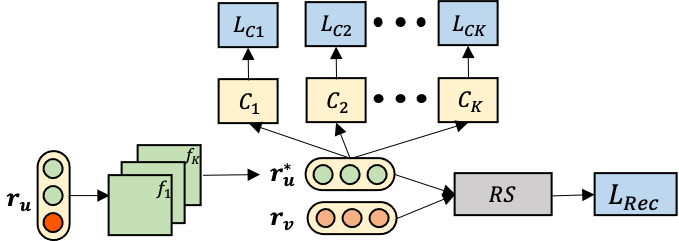}
    \caption{The architecture of our framework. For a given user $u$ with the original representation $\textbf{r}_u$, we first use the filter module to remove the information about sensitive features and get the filtered embedding $\textbf{r}_u^*$. Then we use the corresponding classifier $C_k$ to predict the $k$-th sensitive feature from the filtered embedding, and on the other hand, we train a recommender system (RS) for the main task. The loss of recommendation and classification are optimized together.}
    \label{fig3}
    \vspace{-10pt}
\end{figure}

\subsubsection{\textbf{Filter Module}}

Given a learning algorithm that learns user embedding $\textbf{r}_u$ to generate recommendations for user $u$, we require the embedding $\textbf{r}_u$ to be independent from certain user features to achieve counterfactual fairness. Therefore, 
% We use $u$ to represent user information which can be anything helpful to train recommender models including user ID, user features, and random or pre-trained embeddings, etc. Different models may select different user information to use for training. Suppose $h$ is a encoder to map user information $u$ to a embedding $h(u)$ which is used as the input of the model. 
we first introduce a filter module with a set of filter functions, which are used to filter out the information about certain sensitive features in the user embeddings. We denote the filter function as $f: \mathbb{R}^{d} \mapsto \mathbb{R}^{d}$, and the filtered embedding $f(\textbf{r}_u)$ is independent from certain sensitive features while maintaining other insensitive information of the user. To meet users' personalized fairness demands, we allow each user to select a set of sensitive features $\textbf{Q}_u \subseteq\{1, \ldots, K\}$. To achieve personalized fairness in recommendation, we provide two methods as follows.

In most recommendation scenarios, such as in movie, music, and e-commerce recommendation, users are not willing to share too much personalized information about themselves with the system, and they may only select a few sensitive features for generating fair recommendations, i.e., $K$ will be a very small number. In this scenario, a straightforward way is to train one filter function for each potential combination of the sensitive features. For example, if $K=2$, and $\textbf{Q}_u$ contains two sensitive features \textit{Age} and \textit{Gender}, then we need to train filter functions $f_A, f_G, f_{A,G}$ to remove the sensitive information of \textit{Age}, \textit{Gender}, and both \textit{Age} and \textit{Gender}, from the user embeddings, respectively. We call this method the Separate Method (SM). The architecture of separate method of this example is show in Figure.\ref{fig:pic2}.  We denote the filtered embedding of user $u$ in terms of the selected sensitive feature set $\textbf{Q}_u$ as follows.
$$
\textbf{r}_u^*=f_{\textbf{Q}_u}(\textbf{r}_u)
$$

\begin{figure}[t]
    \centering
    \includegraphics[scale=0.75]{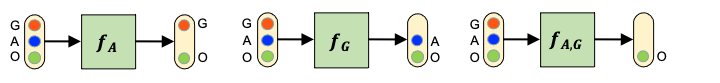}
    \caption{The architecture of Separate Method. G represents \textit{Gender}; A represents \textit{Age}; O represents \textit{Others}.}
    \vspace{-5pt}
    \label{fig:pic2}
\end{figure}

\begin{figure}[t]
    \centering
    \includegraphics[scale=0.75]{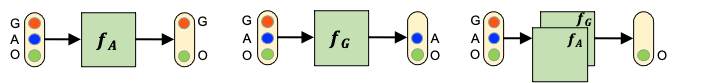}
    \caption{The architecture of Combination Method. G represents \textit{Gender}; A represents \textit{Age}; O represents \textit{Others}.}
    \vspace{-10pt}
    \label{fig:pic3}
\end{figure}
% $Q \subseteq\mathcal{K}=\{1, \ldots, K\}.$ 

However, in some cases such as social network recommendation, users may have many sensitive features to consider, and the potential combinations of sensitive features will be quite a lot. Under such scenarios, training one filter for each combination is infeasible due to the exponential number of combinations. Therefore, we introduce the Combination Method (CM) to achieve personalized fairness in recommendation. Specifically, we train one filter function for each sensitive feature. The filter functions $\{f_1,f_2,...,f_K\}$ correspond to sensitive features $\{Z_1,Z_2,...,Z_K\}$, where $f_k$ is trained to filter the information about $Z_k$. Considering the above example where $\textbf{Q}_u$ contains two sensitive features \textit{Age} and \textit{Gender}, the architecture of combination method is shown in Figure.\ref{fig:pic3}. To 
generate the feature independent embedding of user $u$ with $\textbf{Q}_u=\{q_1,q_2,\ldots,q_{|\textbf{Q}_u|}\}$, we can simply combine the $|\textbf{Q}_u|$ filtered embeddings as:
$$
\textbf{r}_u^*=G(f_{q_1}(\textbf{r}_u),f_{q_2}(\textbf{r}_u),\ldots,f_{q_{|\textbf{Q}_u|}}(\textbf{r}_u))
$$
where $G$ is a combination function that takes the $|\textbf{Q}_u|$ filtered embeddings as input, and outputs the embedding which is independent from all the features in $\textbf{Q}_u$ without changing the embedding dimension. For example, we can simply use the average of the $|\textbf{Q}_u|$ filtered embeddings as $G$. We will compare the performance of the SM and CM methods in the experiment section. 

% $$
% \mathrm{F}(u, Q)=\frac{1}{|Q|} \sum_{k \in Q} f_{k}(h(u))
% $$

Furthermore, for non-personalized situation where the fairness demands of users are the same---e.g., all the users ask for fair recommendations on race---we can simply train one filter function corresponding to the sensitive features.

\subsubsection{\textbf{Discriminator Module}}
To learn filter functions, we use the idea of adversary learning to train a set of discriminators. Specifically, for each sensitive feature $Z_k$, we train a classifier $C_{k}: \mathbb{R}^{d} \mapsto$ [0,1], which attempts to predict $Z_k$ from the user embeddings. The goal of the filter functions is to make it hard for the classifiers to predict the sensitive features from the user embeddings, while the goal of the discriminators is to fail the filter functions. Concretely, the training process tries to jointly optimize both goals.

\subsubsection{\textbf{Adversary Training}}

We use $\mathcal{L}_{Rec}$ to denote the loss of the recommendation task. Depending on the recommendation model, $\mathcal{L}_{Rec}$ can be the pair-wise ranking loss \cite{rendle2012bpr} or mean square error loss \cite{koren2009matrix}, etc. We use $\mathcal{L}_C$ to denote the loss of the discriminators, i.e., the loss of the classification task, which is a cross-entropy loss in our implementation. We thus define our adversary learning loss as follows:
\begin{equation}\label{eq}
\mathcal{L}=\sum_{u, v, \textbf{Q}_u}\left(\mathcal{L}_{\text {Rec}}(u,v,\textbf{Q}_u)
-\lambda \sum_{k \in \textbf{Q}_u} \sum_{z_{k} \in Z_{k}} \mathcal{L}_C \left(\textbf{r}_u^*, z_{k}\right)\right)
\end{equation}
where adversarial coefficient $\lambda$ controls the trade-off between recommendation performance and fairness. We study the influence of $\lambda$ in the ablation study section. The adversary learning algorithm is shown in Algorithm \ref{alg}.

We also provide the following theorem to show the theoretical guarantee that the adversary training procedure can make the filtered user embeddings independent from the sensitive features.

\begin{theorem}
\label{theorem1}
If (1) the filter functions and discriminators are implemented with sufficient capacity, and (2) at each step of Algorithm~\ref{alg}, the discriminators are allowed to reach their optimum given the filter functions, and (3) the filter functions are optimized according to the loss function with discriminators fixed, then we have for any user $u$, $\textbf{r}_u \perp \textbf{Z}_{u}$ as $\lambda\rightarrow \infty$. 
\end{theorem}

\begin{proof}
For simplicity, we consider the case of a single binary sensitive feature, i.e., we have one sensitive feature $Z$ which can be 0 or 1. In this case, the loss of the discriminators in Eq.\eqref{eq}---which will be a binary cross-entropy loss---is the same as the loss of Generative Adverserisal Networks (GAN) \cite{goodfellow2014generative}. According to the Proposition 2 of \cite{goodfellow2014generative}, it has been proven that if (1) the generator and discriminator have enough capacity, (2) the discriminator is allowed to reach the optimum during the training process, and (3) the generator is updated with the discriminator fixed so as to improve the criterion, then the distribution of the fake data will converge to the distribution of the real data. As a result, when classifying a binary sensitive feature, we have that the distributions of user embeddings with sensitive feature $Z=0$ and $Z=1$ will be the same once converged, i.e., we have $P(\textbf{r}_u^*|Z=0)=P(\textbf{r}_u^*|Z=1)$, which gives $\textbf{r}_u^* \perp Z$.
\end{proof}

The theoretical intuition above can be generalized to the multi-class and multi-feature settings. \citeauthor{cho2020fair} \cite{cho2020fair} have theoretically shown how to optimize the independence between predictions and sensitive features in multi-settings from a mutual information perspective, which leads to the same result in Theorem \ref{theorem1}.

\subsection{Training Algorithm}

For adversary learning, we adopt mini-batch training in our implementation. Specifically, for each batch, we first feed the input to the model to obtain $\mathcal{L}_{Rec}$ and $\mathcal{L}_{C}$, and then we fix the parameters in the discriminator and optimize the recommendation model as well as the corresponding filter functions by minimizing $\mathcal{L}$. After that, the parameters of the recommendation model and filter functions are fixed, and $\mathcal{L}_C$ is minimized for $t$ steps. Here $t=10$ in our implementation. 

To achieve personalized fairness, we allow each user to select a set of concerned sensitive features $\textbf{Q}_u \subseteq\{1, \ldots, K\}.$ In implementation of the learning algorithm, we sample a binary mask for each batch to determine $\textbf{Q}_u$ to train the filtered embedding $\textbf{r}_u^*$. Specifically, the binary mask is sampled from $K$ independent Bernoulli distributions with the probability $p=0.5$ under the assumption that there is no causal relation between sensitive features when user selects them. When there is a need to consider the dependence between sensitive features, we can simply apply other distributions based on different applications. Again, the pseudo-code for the entire training algorithm is given in Algorithm \ref{alg}.

\begin{algorithm}[t]\label{alg}
\caption{Adversarial Training Algorithm}
\label{alg}
\SetKwInOut{Input}{Input}
\SetKwInOut{Output}{Output}
\SetAlgoLined
\Input{Training user set $\mathcal{U}$; item set $\mathcal{V}$; Recommendation model RS; Filter functions $\mathcal{F}$; Discriminators $\mathcal{C}$; Sensitive features $\textbf{Z}$; Training epochs $M$; Discriminator training steps $T$; Adversarial coefficient $\lambda$;\\}
% \Output{Latent logicalized vector representations $\mathbf{e}_v, \forall v\in \mathcal{V}$}
\BlankLine
Initialize: user embeddings $\mathbf{r}_u, \forall u\in \mathcal{U}$, item embeddings $\mathbf{r}_v, \forall v\in \mathcal{V}$\;
\For{$epoch \leftarrow 1$ \KwTo $M$} {
    \For{$u \in \mathcal{U}$, $v \in \mathcal{V}$} {
        $\mathbf{Q} \leftarrow$ sample binary filter mask\;
        \If{Separate Method} {
            $F \leftarrow$ get filter function $f_{\textbf{Q}}$ from $\mathcal{F}$\;
        } 
        \If {Combination Method} {$F \leftarrow$ get filter functions $\{f_k\}_{k\in \textbf{Q}}$ from $\mathcal{F}$\;
        }
        $\{C_k\}_{k\in \textbf{Q}} \leftarrow$ get discriminators from $\mathcal{C}$\;
        $\{z_k\}_{k\in \textbf{Q}} \leftarrow$ get feature values from $\textbf{Z}$\;
        
        $\textbf{r}_u^* \leftarrow F(\textbf{r}_u)$\Comment{obtain filtered user embedding}\;
        $\mathcal{L}_{Rec} \leftarrow$ RS($\mathbf{r}_u^*, \mathbf{r}_v$)\;
        $\mathcal{L}_C \leftarrow \sum_{k\in \textbf{Q}}C_k(\mathbf{r}_u^*, z_k)$ 
        
        $\mathcal{L} \leftarrow \mathcal{L}_{Rec} + \lambda \mathcal{L}_C$\;
        Optimize $\mathcal{L}$ $w.r.t$ $\mathbf{r}_u, \mathbf{r}_v, F, \mathrm{RS}$, with $\{C_k\}_{k\in \textbf{Q}}$ fixed\;
        \For{$t\leftarrow 1$ \KwTo $T$} {
            $\mathcal{L}_C\leftarrow \sum_{k\in \textbf{Q}}C_k(\mathbf{r}_u^*, z_k)$\;
            Optimize $\mathcal{L}_C$ $w.r.t$ $\{C_k\}_{k\in \textbf{Q}}$ with $\mathbf{r}_u, F$ fixed\;
        }
    }
}

\end{algorithm}

\section{Experiments}
\label{sec:experiments}

In this section, we first briefly introduce the datasets, baselines and experimental setup used for the experiments. Then we show and analyze the main experimental results, including the comparison of recommendation performance and fairness between the baseline models and the two fairness models. Finally, we conduct ablation studies to further analyze the algorithm. 

\subsection{\textbf{Dataset Description}}

To evaluate the models under different data scales, data sparsity and application scenarios, we perform experiments on a movie recommendation dataset and an insurance recommendation dataset, which are two real-world and publicly available datasets. 

\subsubsection*{\normalfont \textbf{MovieLens\footnote{https://grouplens.org/datasets/movielens/1m/}}}We use the MovieLens-1M dataset which contains user-item interactions and user profile information for movie recommendation. We select \textit{gender}, \textit{age} and \textit{occupation} as user sensitive features, where \textit{gender} is a binary feature, \textit{occupation} is a 21-class feature, and for \textit{age}, users are assigned into 7 groups based on their age range. 

\subsubsection*{\normalfont \textbf{Insurance}\footnote{https://www.kaggle.com/mrmorj/insurance-recommendation}}This is a Kaggle dataset with the goal of recommending insurance products to a target user. For each user, we select \textit{gender}, \textit{marital\_status} and \textit{occupation} as sensitive features, where \textit{gender} is still a binary feature. For $\textit{marital\_status}$ and \textit{occupation}, we group up the minority classes to transform them into 3-class features due to the severe data imbalance over classes. To guarantee the data quality for training models, we filter out the users with less than 4 interactions to make a denser dataset.

The statistics of the datasets are summarized in Table~\ref{tab:freq}. In our experiments, we split each dataset into train ($80\%$), validation ($10\%$) and test sets ($10\%$) and all the baseline models share these datasets for training and evaluation.

\subsection{\textbf{Evaluation Methods}}
We consider standard metrics Normalized Discounted Cumulative Gain ($\mathrm{NDCG}@N$) and Hit rate ($\mathrm{Hit}@ N$) scores to evaluate the top-$N$ recommendation quality. For the MovieLens dataset, we report NDCG@5 and Hit@5. For the Insurance dataset, we show NDCG@3 and Hit@3 scores due to the limited number of candidates in this dataset. For efficiency consideration, we use sampled negative interactions for evaluation instead of computing the user-item pair scores for each user over the entire item space \cite{zhao2020revisiting}. For each user, we randomly select 100 negative samples that the user has never interacted with. These negative items are put together with the positive item in the validation or test set to constitute the user's candidates list. Then we compute the metric scores over this candidates list to evaluate the recommendation model's top-$N$ ranking performance. The result of all metrics in our experiments are averaged over all users.

Following the settings in learning fair representations by adversary learning such as \cite{elazar2018adversarial,bose2019compositional,arduini2020adversarial}, to evaluate the effectiveness of discriminators, we train a set of attackers which have totally the same structure and capacity as the discriminators. Specifically, after we finish training the main algorithm, we input the filtered user embeddings and their corresponding sensitive labels to the attackers so as to train them to classify the sensitive features from the filtered embeddings. Just as the discriminators, we train one attacker for each sensitive feature. If the attackers can distinguish sensitive features from the user embeddings, then we say that sensitive features are leaked into the user embeddings, thus the recommendation model is not counterfactually fair. 

For training and evaluating attackers, we split the data into train (80\%) and test sets (20\%). We report AUC score for each attacker to show if the filtered user embeddings can be classified correctly by the attacker. For multi-class evaluation, we calculate the AUC score for all the combination of feature pairs and apply their macro-average to make the result insensitive to imbalanced data. The AUC score falls into the range of [0.5,1], the lower the better. An ideal result to meet the counterfactual fairness requirement is an AUC score of about 0.5, which means the attacker cannot guess the sensitive feature out of the user embeddings at all.

\subsection{Baselines}
In order to evaluate the effectiveness of our proposed framework, we apply our method over both shallow and deep recommendation models. We introduce the baseline models as follows:
\begin{itemize}
    \item \textbf{PMF} \cite{mnih2008probabilistic}: The Probabilistic Matrix Factorization algorithm by adding Gaussian prior into the user and item latent factor distributions for matrix factorization. 
    \item \textbf{BiasedMF} \cite{koren2009matrix}: A matrix factorization algorithm which takes user and item latent factors as well as the global bias terms into consideration. 
    \item \textbf{DeepModel} \cite{cheng2016wide}: This algorithm applies deep neural network with non-linear activation functions to train a user and item matching function.
    \item \textbf{DMF} \cite{xue2017deep}: Deep Matrix Factorization is a deep model for recommendation, which uses multi-layer perceptron with non-linear activation function to encode the raw user-item interaction matrix into dense latent factor representations.
\end{itemize}

\begin{table}[t!]
  \caption{Statistics of the datasets}
  \label{tab:freq}
  \begin{tabular}{lrrrr}
    \toprule
    Dataset& \#Interactions & \#Users & \#Items & Sparsity  \\
    \midrule
     MovieLens & 1,000,209& 6,040& 3,952&95.81\%\\
     Insurance &5,382 &1,231 &21 &79.18\%\\
  \bottomrule
 \label{data}
\end{tabular}
\vspace{-10pt}
\end{table}

\begin{table}[t!]
\setlength{\tabcolsep}{2pt}
  \centering
  \caption{AUC scores of all attackers on the MolveLens and Insurance datasets. G, A, O represent $gender$, $age$ and $occupation$ respectively on MovieLens; while G, M, O represent $gender$, $marital\_status$ and $occupation$ respectively on Insurance. The best results are highlighted in bold.}
  \resizebox{1\linewidth}{!}{
    \begin{tabular}{clrrrrrr}
    \toprule
          \multirow{2}{*}{}& & \multicolumn{3}{c}{MoiveLens} & \multicolumn{3}{c}{Insurance} \\
          \cmidrule(lr){3-5}
          \cmidrule(lr){6-8}
          &  &AUC-G &AUC-A &AUC-O &AUC-G &AUC-M &AUC-O\\
    \midrule
    \multirow{3}[2]{*}{PMF} & Orig. & 0.7697 & 0.8428 & 0.6024 & 0.6253 & 0.7098 & 0.6577 \\
          & SM    & \textbf{0.5389} & \textbf{0.5560} & \textbf{0.5289} & \textbf{0.5340} & \textbf{0.5377} & \textbf{0.5492} \\
          & CM    & 0.5532 & 0.5951 & 0.5396 & 0.5419 & 0.5789 & 0.5540 \\
    \midrule
    \multirow{3}[2]{*}{BiasedMF} & Orig. & 0.7870 & 0.8403 & 0.6064 & 0.6183 & 0.7715 & 0.6357 \\
          & SM    & \textbf{0.5345} & \textbf{0.5601} & \textbf{0.5258} & \textbf{0.5000 }  & \textbf{0.5405} & \textbf{0.5555 }\\
          & CM    & 0.5519 & 0.5757 & 0.5300  & 0.5491 & 0.5430 & 0.5717 \\
    \midrule
    \multirow{3}[2]{*}{DeepModel} & Orig. & 0.7165 & 0.7571 & 0.5481 & 0.5952 & 0.6339 & 0.6086 \\
          & SM    & 0.5545 & \textbf{0.5833} & 0.5445 & \textbf{0.5202} & \textbf{0.5687} & 0.5815 \\
          & CM    & \textbf{0.5371} & 0.6075 & \textbf{0.5247} & 0.5335 & 0.5765 & \textbf{0.5407} \\
    \midrule
    \multirow{3}[2]{*}{DMF} & Orig. & 0.7049 & 0.7238 & 0.5710 & 0.6172 & 0.6309 & 0.6023 \\
          & SM    & 0.6073 & 0.5670 & 0.5289 & 0.5421 & \textbf{0.5638} & \textbf{0.5653}\\
          & CM    & \textbf{0.5000}   & \textbf{0.5297} & \textbf{0.5120} & \textbf{0.5258} & 0.5873 & 0.5791 \\
    \bottomrule
    \end{tabular}%
    }
  \label{auc}%
  \vspace{-10pt}
\end{table}%

% Table generated by Excel2LaTeX from sheet 'Sheet4'
\begin{table*}[t!]
\setlength{\tabcolsep}{2pt}
  \centering
  \caption{The recommendation performance of baselines, Separate Method (SM) and Combination Method (CM) on MovieLens. Orig. represents the baseline model; G, A, O represent $gender$, $age$ and $occupation$, respectively. The performance of SM and CM are evaluated for all combinations of the three sensitive features. For example, SM-G represents the performance of filtering user embeddings by $f_{gender}$ trained via SM. The better results between SM and CM are highlighted in bold.}
    \resizebox{1\textwidth}{!}{
    \begin{tabular}{clrrrrrrrrrrrrrrr}
    \toprule
          &       & \multicolumn{1}{c}{Orig.} & \multicolumn{1}{c}{SM-G} & \multicolumn{1}{c}{CM-G} & \multicolumn{1}{c}{SM-A} & \multicolumn{1}{c}{CM-A} & \multicolumn{1}{c}{SM-O} & \multicolumn{1}{c}{CM-O} & \multicolumn{1}{c}{SM-GA} & \multicolumn{1}{c}{CM-GA} & \multicolumn{1}{c}{SM-GO} & \multicolumn{1}{c}{CM-GO} & \multicolumn{1}{c}{SM-AO} & \multicolumn{1}{c}{CM-AO} & \multicolumn{1}{c}{SM-GAO} & \multicolumn{1}{c}{CM-GAO} \\
    \midrule
    \multirow{2}[1]{*}{PMF} & N@5 & 0.4961 & \textbf{0.4801} & 0.4604 & \textbf{0.4781} & 0.4596 & \textbf{0.4751} & 0.4605 & \textbf{0.4730} & 0.4673 & \textbf{0.4737} & 0.4685 & 0.4674 & \textbf{0.4681} & 0.4599 & \textbf{0.4705} \\
          & H@5 & 0.6493 & \textbf{0.6342} & 0.6179 & \textbf{0.6318} & 0.6188 & \textbf{0.6291} & 0.6191 & \textbf{0.6273} & 0.6251 & \textbf{0.6282} & 0.6274 & 0.6207 & \textbf{0.6265} & 0.6165 & \textbf{0.6289} \\
\midrule
    \multirow{2}[1]{*}{BiasedMF} & N@5 & 0.4960 & \textbf{0.4776} & 0.4649 & \textbf{0.4740} & 0.4665 & \textbf{0.4748} & 0.4672 & 0.4710 & \textbf{0.4732} & 0.4699 & \textbf{0.4742} & 0.4672 & \textbf{0.4744} & 0.4573 & \textbf{0.4767} \\
          & H@5 & 0.6471 & \textbf{0.6305} & 0.6205 & \textbf{0.6270} & 0.6220 & \textbf{0.6285} & 0.6233 & 0.6248 & \textbf{0.6291} & 0.6240 & \textbf{0.6305} & 0.6208 & \textbf{0.6303} & 0.6118 & \textbf{0.6324} \\
    \midrule
    \multirow{2}[2]{*}{DeepModel} & N@5 & 0.3935 & \textbf{0.3834} & 0.3803 & \textbf{0.3827} & 0.3793 & \textbf{0.3825} & 0.3790 & \textbf{0.3819} & 0.3809 & \textbf{0.3820} & 0.3808 & 0.3797 & \textbf{0.3800} & 0.3782 & \textbf{0.3808} \\
          & H@5 & 0.5501 & \textbf{0.5370} & 0.5338 & \textbf{0.5357} & 0.5325 & \textbf{0.5357} & 0.5325 & \textbf{0.5349} & 0.5343 & \textbf{0.5350} & 0.5343 & 0.5322 & \textbf{0.5337} & 0.5311 & \textbf{0.5344} \\
    \midrule
    \multirow{2}[2]{*}{DMF} & N@5 & 0.3307 & \textbf{0.3262} & 0.3167 & \textbf{0.3256} & 0.3168 & \textbf{0.3260} & 0.3166 & \textbf{0.3254} & 0.3177 & \textbf{0.3267} & 0.3169 & \textbf{0.3253} & 0.3164 & \textbf{0.3263} & 0.3183 \\
          & H@5 & 0.4795 & \textbf{0.4731} & 0.4598 & \textbf{0.4709} & 0.4588 & \textbf{0.4731} & 0.4603 & \textbf{0.4707} & 0.4606 & \textbf{0.4714} & 0.4599 & \textbf{0.4714} & 0.4603 & \textbf{0.4732} & 0.4622 \\
    \bottomrule
    \end{tabular}%
    }
  \label{tb:movie_result}%
\end{table*}%

% Table generated by Excel2LaTeX from sheet 'Sheet5'
\begin{table*}[t!]
\setlength{\tabcolsep}{2pt}
  \centering
  \caption{The recommendation performance of baselines, Separate Method (SM) and Combination Method (CM) on Insurance. Orig. represents the baseline model; G, M, O represent $gender$, $marital\_status$ and $occupation$, respectively. The performance of SM and CM are evaluated for all combinations of the three sensitive features. For example, SM-G represents the performance of filtering user embeddings by $f_{gender}$ trained via SM. The better results between SM and CM are highlighted in bold.}
  \resizebox{1\textwidth}{!}{

    \begin{tabular}{clrrrrrrrrrrrrrrr}
    \toprule
          &       & \multicolumn{1}{c}{Orig.} & \multicolumn{1}{c}{SM-G} & \multicolumn{1}{c}{CM-G} & \multicolumn{1}{c}{SM-M} & \multicolumn{1}{c}{CM-M} & \multicolumn{1}{c}{SM-O} & \multicolumn{1}{c}{CM-O} & \multicolumn{1}{c}{SM-GM} & \multicolumn{1}{c}{CM-GM} & \multicolumn{1}{c}{SM-GO} & \multicolumn{1}{c}{CM-GO} & \multicolumn{1}{c}{SM-MO} & \multicolumn{1}{c}{CM-MO} & \multicolumn{1}{c}{SM-GMO} & \multicolumn{1}{c}{CM-GMO} \\
    \midrule
    \multirow{2}[1]{*}{PMF} & N@3 & 0.6518 & \textbf{0.6528} & 0.6208 & \textbf{0.6521} & 0.6081 & \textbf{0.6406} & 0.6208 & \textbf{0.6242} & 0.6173 & \textbf{0.6535} & 0.6208 & \textbf{0.6228} & 0.6163 & \textbf{0.6390} & 0.6191 \\
          & H@3 & 0.7528 & \textbf{0.7398} & 0.7026 & \textbf{0.7398} & 0.6914 & \textbf{0.7416} & 0.7026 & \textbf{0.7323} & 0.6989 & \textbf{0.7398} & 0.7026 & \textbf{0.7435} & 0.6989 & \textbf{0.7305} & 0.7007 \\
    \midrule
    \multirow{2}[1]{*}{BiasedMF} & N@3 & 0.6209 & 0.5936 & \textbf{0.6095} & \textbf{0.6190} & 0.5995 & 0.6031 & \textbf{0.6112} & 0.5936 & \textbf{0.6079} & 0.5991 & \textbf{0.6116} & 0.6041 & \textbf{0.6056} & 0.6028 & \textbf{0.6095} \\
          & H@3 & 0.7082 & 0.6803 & \textbf{0.6952} & \textbf{0.7100} & 0.6822 & 0.6877 & \textbf{0.6952} & 0.6803 & \textbf{0.6952} & 0.6803 & \textbf{0.6970} & \textbf{0.6952} & 0.6914 & 0.6859 & \textbf{0.6952} \\
    \midrule
    \multirow{2}[2]{*}{DeepModel} & N@3 & 0.6438 & \textbf{0.6389} & 0.6315 & \textbf{0.6359} & 0.6290 & \textbf{0.6410} & 0.6317 & \textbf{0.6355} & 0.6250 & \textbf{0.6333} & 0.6190 & \textbf{0.6401} & 0.6317 & \textbf{0.6357} & 0.6275 \\
          & H@3 & 0.7398 & \textbf{0.7212} & 0.7175 & \textbf{0.7193} & 0.7082 & \textbf{0.7286} & 0.7193 & \textbf{0.7249} & 0.7063 & \textbf{0.7138} & 0.7082 & \textbf{0.7268} & 0.7193 & \textbf{0.7212} & 0.7156 \\
    \midrule
    \multirow{2}[2]{*}{DMF} & N@3 & 0.5301 & \textbf{0.4988} & 0.4751 & \textbf{0.5122} & 0.4927 & \textbf{0.5108} & 0.4858 & \textbf{0.5143} & 0.4731 & \textbf{0.5115} & 0.4827 & 0.4986 & \textbf{0.5040} & \textbf{0.5231} & 0.4652 \\
          & H@3 & 0.6822 & 0.6283 & \textbf{0.6301} & \textbf{0.6468} & 0.6431 & \textbf{0.6413} & 0.6320 & \textbf{0.6375} & 0.6115 & \textbf{0.6394} & 0.6245 & \textbf{0.6264} & 0.6245 & \textbf{0.6580} & 0.6190 \\
    \bottomrule
    \end{tabular}%
    }
  \label{tb:insurance_result}
  \vspace{-5pt}
\end{table*}%

\subsection{Experimental Settings}
To better accommodate the ranking task, we apply the Bayesian Personalized Ranking (BPR)~\cite{rendle2012bpr} loss as the recommendation loss in Eq.\eqref{eq} for all the baseline models. For each user-item pair in the training dataset, we randomly sample one item that the user has never interacted with as the negative sample in one training epoch. We set the learning rate to 0.001.  $\ell_2$-regularization coefficient is 0.0001 for all the datasets. Dropout rate is 0.2. Early stopping is applied and the best models are selected based on the performance on the validation set. Rectified Linear Unit (ReLU) is used as the activation function for DMF and DeepModel. We apply Adam~\cite{kingma2014adam} as the optimization algorithm to update the model parameters. The adversarial coefficient $\lambda$ in Eq.\eqref{eq} is selected from $[10, 20, 50]$ for MovieLens, while $[100, 200, 500, 1000]$ for the Insurance dataset. The filter modules are two-layer neural networks with LeakyReLU as the non-linear activation function. The classifiers (discriminators and attackers) are multi-layer perceptrons with the number of layers set to 7, LeakyReLU as the activation function, and the dropout rate is set to 0.3. Batch normalization is applied for training classifiers. The open source code can be found in GitHub~\footnote{\url{https://github.com/yunqi-li/Personalized-Counterfactual-Fairness-in-Recommendation}}.
%\footnote{Our code will be released after the paper is published}.

\subsection{Main Results}

In this section, we show the main results of our experiments, including comparing recommendation performance and fairness for all the baseline models under the SM and CM settings. 
% \subsubsection*{\normalfont \textbf{Fairness Improvement}} 
\\\\\noindent
\textbf{Fairness Improvement}. We provide the evaluation results of attackers in Table~\ref{auc} to present the effectiveness of adversary training for achieving counterfactual fairness. According to the table, we observe that the AUC scores of the baseline models are significantly higher than 0.5, which means that the attackers can easily discriminate user embeddings from sensitive features. In other words, for general recommendation models, no matter they are shallow or deep, the sensitive information of users can be learned from the data even such information is not explicitly used, thus leads to unfair recommendation results. Furthermore, the AUC scores of both SM and CM methods are around 0.5, i.e., it is hard for the attackers to distinguish the sensitive features from the filtered user embeddings. It indicates that both CM and SM implementations are effective for achieving counterfactual fairness in recommendations. Addtionally, from the observations of the AUC scores for SM and CM methods, we see that SM achieves lower AUC than CM in most cases, which shows that training separate filter functions for each combination of sensitive features can usually remove more sensitive information than the combination method.
% \subsubsection*{\normalfont \textbf{Recommendation Performance}}
\\\\\noindent
\textbf{Recommendation Performance}. We compare the recommendation performance of the baseline models and the two fair methods. We show $\mathrm{NDCG}$ and $\mathrm{Hit}$ results of all the models on the MoiveLens and Insurance datasets in Table~\ref{tb:movie_result} and Table~\ref{tb:insurance_result}, respectively. We can see that both of the two fair methods can still achieve high recommendation quality. Although fair methods will suffer from a little sacrifice on recommendation performance to guarantee the fairness requirement, the recommendation performance is still very close to the original performance. This is acceptable as there is typically an inevitable trade-off between prediction accuracy and fairness \cite{menon2018cost, zafar2017fairness, bose2019compositional}. The reason why there exists trade-off between fairness and recommendation performance is that the fair methods are aiming at filtering out the information of certain sensitive features from user embeddings, which will to some extent reduce the information contained in the embeddings, thus decreasing the recommendation performance.

 %\subsubsection*{\normalfont \textbf{Recommendation Performance}} We show $\mathrm{NDCG}$ and $\mathrm{Hit}$ results of all the models on the MoiveLens and Insurance datasets in Table~\ref{tb:movie_result} and Table~\ref{tb:insurance_result}, respectively. The reported results involves both SM and CM methods. Although fair methods will suffer from sacrificing recommendation performance to guarantee the fairness requirement, it is acceptable as there is typically an inevitable trade-off between prediction accuracy and fairness \cite{menon2018cost, zafar2017fairness, bose2019compositional}. We can see that both of the two fair methods can still achieve high recommendation quality. The reason why there exists trade-off between fairness and recommendation performance is that the fair methods are aiming at filtering out the information of certain sensitive features from user embeddings, which will to some extent reduce the information contained in embeddings, thus decreasing the recommendation performance. 

% \toprule
% \multirow{2}{*}{}& & &\multicolumn{4}{c}{\textbf{Beauty}} & \multicolumn{4}{c}{\textbf{Grocery}} & \multicolumn{4}{c}{\textbf{Health}} \\ 
% \cmidrule(lr){4-7}
% \cmidrule(lr){8-11}
% \cmidrule(lr){12-15}
% & & & Overall &Adv. &Disadv. &UGF &Overall &Adv. &Disadv. & UGF &Overall &Adv. &Disadv. &UGF \\
% \midrule

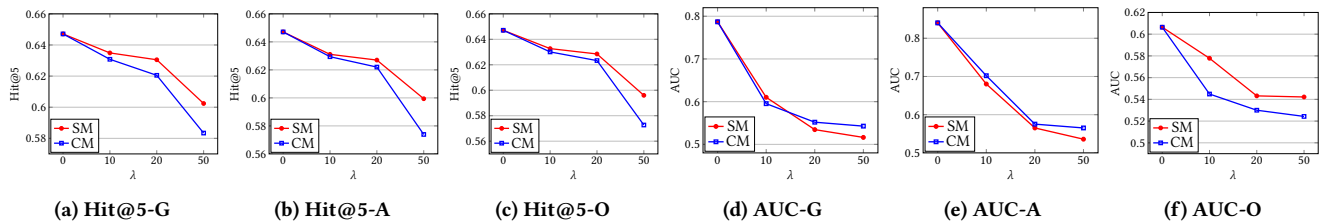
\begin{figure*}[t!]
    \centering
    % \caption[]
    %     {Impact of the adversarial coefficient $\lambda$ on (a)-(c) recommendation performance $w.r.t$ Hit@5, and (d)-(f) classification quality of the attackers $w.r.t$ AUC. In the subfigure titles, G, A and O represent the sensitive features $gender$, $age$ and $occupation$, respectively. The reported results are from the Biased-MF model on the MovieLens dataset.} 
        
        \pgfplotsset{label style={font=\huge},
                    tick label style={font=\huge}}
        \begin{subfigure}[b]{0.16\textwidth}   
            \centering 
            \resizebox{\linewidth}{!}{
                \begin{tikzpicture}
                    \begin{axis}[
                        scaled y ticks = false,
                        ymin=0.57, ymax=0.66,
                        ylabel={Hit@5},
                        xlabel={$\lambda$},
                        symbolic x coords={0,10,20, 50},
                        xtick=data,
                        scaled ticks=false, tick label style={/pgf/number format/fixed},
                        ymajorgrids=true,
                        legend style={at={(0.01,0.01)},
                        anchor=south west,
                        font=\fontsize{18}{5}\selectfont,
                        }
                        % legend style={
                        % draw=none, % ?
                        % text depth=0pt,
                        % at={(0.15,-0.30)},
                        % anchor=north west,
                        % legend columns=-1, 
                        % font=\fontsize{8}{5}\selectfont,
                        % % default spacing:
                        % column sep=0.1cm,
                        % % The text "Legend:"
                        % % /tikz/column 2/.style={column sep=0pt,font=\bfseries},
                        % %
                        % % the space between legend image and text:
                        % /tikz/every odd column/.append style={column sep=0cm},
                        % }
                    ]
                    \addplot[mark=*,red, line width=1.5pt]
                    coordinates{
                        (0,0.6471)
                        (10,0.6349)
                        (20,0.6305)
                        (50,0.6024)
                    }; 
                    \addlegendentry{SM}
                    \addplot[mark=square, blue, line width=1.5pt]
                    coordinates{
                        (0,0.6471)
                        (10,0.6308)
                        (20,0.6205)
                        (50,0.5834)
                    }; 
                    \addlegendentry{CM} 
                    \end{axis}
                \end{tikzpicture}
            }
            \caption[]%
            {{\small Hit@5-G}}
            % {{\small}}
            \label{fig:hit_g}
        \end{subfigure}
        \begin{subfigure}[b]{0.16\textwidth}   
            \centering 
            \resizebox{\linewidth}{!}{
                \begin{tikzpicture}
                    \begin{axis}[
                        scaled y ticks = false,
                        ymin=0.56, ymax=0.66,
                        ylabel={Hit@5},
                        xlabel={$\lambda$},
                        symbolic x coords={0,10,20, 50},
                        xtick=data,
                        scaled ticks=false, tick label style={/pgf/number format/fixed},
                        ymajorgrids=true,
                        legend style={at={(0.01,0.01)},
                        anchor=south west,
                        font=\fontsize{18}{5}\selectfont,
                        }
                        % legend style={
                        % draw=none, % ?
                        % text depth=0pt,
                        % at={(0.15,-0.30)},
                        % anchor=north west,
                        % legend columns=-1, 
                        % font=\fontsize{8}{5}\selectfont,
                        % % default spacing:
                        % column sep=0.1cm,
                        % % The text "Legend:"
                        % % /tikz/column 2/.style={column sep=0pt,font=\bfseries},
                        % %
                        % % the space between legend image and text:
                        % /tikz/every odd column/.append style={column sep=0cm},
                        % }
                    ]
                    \addplot[mark=*,red, line width=1.5pt]
                    coordinates{
                        (0,0.6471)
                        (10,0.6311)
                        (20,0.6270)
                        (50,0.5994)
                    }; 
                    \addlegendentry{SM}
                    \addplot[mark=square, blue, line width=1.5pt]
                    coordinates{
                        (0,0.6471)
                        (10,0.6294)
                        (20,0.6220)
                        (50,0.5739)
                    }; 
                    \addlegendentry{CM} 
                    \end{axis}
                \end{tikzpicture}
            }
            \caption[]%
            {{\small Hit@5-A}}
            % {{\small}}
            \label{fig:hit_g}
        \end{subfigure}
        \begin{subfigure}[b]{0.16\textwidth}   
            \centering 
            \resizebox{\linewidth}{!}{
                \begin{tikzpicture}
                    \begin{axis}[
                        scaled y ticks = false,
                        ymin=0.55, ymax=0.66,
                        ylabel={Hit@5},
                        xlabel={$\lambda$},
                        symbolic x coords={0,10,20, 50},
                        xtick=data,
                        scaled ticks=false, tick label style={/pgf/number format/fixed},
                        ymajorgrids=true,
                        legend style={at={(0.01,0.01)},
                        anchor=south west,
                        font=\fontsize{18}{5}\selectfont,
                        }
                        % legend style={
                        % draw=none, % ?
                        % text depth=0pt,
                        % at={(0.15,-0.30)},
                        % anchor=north west,
                        % legend columns=-1, 
                        % font=\fontsize{8}{5}\selectfont,
                        % % default spacing:
                        % column sep=0.1cm,
                        % % The text "Legend:"
                        % % /tikz/column 2/.style={column sep=0pt,font=\bfseries},
                        % %
                        % % the space between legend image and text:
                        % /tikz/every odd column/.append style={column sep=0cm},
                        % }
                    ]
                    \addplot[mark=*,red, line width=1.5pt]
                    coordinates{
                        (0,0.6471)
                        (10,0.6327)
                        (20,0.6285)
                        (50,0.5960)
                    }; 
                    \addlegendentry{SM}
                    \addplot[mark=square, blue, line width=1.5pt]
                    coordinates{
                        (0,0.6471)
                        (10,0.6301)
                        (20,0.6233)
                        (50,0.5727)
                    }; 
                    \addlegendentry{CM} 
                    \end{axis}
                \end{tikzpicture}
            }
            \caption[]%
            {{\small Hit@5-O}}
            % {{\small}}
            \label{fig:hit_g}
        \end{subfigure}
        \begin{subfigure}[b]{0.16\textwidth}   
            \centering 
            \resizebox{\linewidth}{!}{
                \begin{tikzpicture}
                    \begin{axis}[
                        scaled y ticks = false,
                        ymin=0.48, ymax=0.82,
                        ylabel={AUC},
                        xlabel={$\lambda$},
                        symbolic x coords={0,10,20, 50},
                        xtick=data,
                        scaled ticks=false, tick label style={/pgf/number format/fixed},
                        ymajorgrids=true,
                        legend style={at={(0.01,0.01)},
                        anchor=south west,
                        font=\fontsize{18}{5}\selectfont,
                        }
                        % legend style={
                        % draw=none, % ?
                        % text depth=0pt,
                        % at={(0.15,-0.30)},
                        % anchor=north west,
                        % legend columns=-1, 
                        % font=\fontsize{8}{5}\selectfont,
                        % % default spacing:
                        % column sep=0.1cm,
                        % % The text "Legend:"
                        % % /tikz/column 2/.style={column sep=0pt,font=\bfseries},
                        % %
                        % % the space between legend image and text:
                        % /tikz/every odd column/.append style={column sep=0cm},
                        % }
                    ]
                    \addplot[mark=*,red, line width=1.5pt]
                    coordinates{
                        (0,0.7870)
                        (10,0.6102)
                        (20,0.5345)
                        (50,0.5163)
                    }; 
                    \addlegendentry{SM}
                    \addplot[mark=square, blue, line width=1.5pt]
                    coordinates{
                        (0,0.7870)
                        (10,0.5953)
                        (20,0.5519)
                        (50,0.5427)
                    }; 
                    \addlegendentry{CM} 
                    \end{axis}
                \end{tikzpicture}
            }
            \caption[]%
            {{\small AUC-G}}
            % {{\small}}
            \label{fig:hit_g}
        \end{subfigure}
        \begin{subfigure}[b]{0.16\textwidth}   
            \centering 
            \resizebox{\linewidth}{!}{
                \begin{tikzpicture}
                    \begin{axis}[
                        scaled y ticks = false,
                        ymin=0.50, ymax=0.88,
                        ylabel={AUC},
                        xlabel={$\lambda$},
                        symbolic x coords={0,10,20, 50},
                        xtick=data,
                        scaled ticks=false, tick label style={/pgf/number format/fixed},
                        ymajorgrids=true,
                        legend style={at={(0.01,0.01)},
                        anchor=south west,
                        font=\fontsize{18}{5}\selectfont,
                        }
                        % legend style={
                        % draw=none, % ?
                        % text depth=0pt,
                        % at={(0.15,-0.30)},
                        % anchor=north west,
                        % legend columns=-1, 
                        % font=\fontsize{8}{5}\selectfont,
                        % % default spacing:
                        % column sep=0.1cm,
                        % % The text "Legend:"
                        % % /tikz/column 2/.style={column sep=0pt,font=\bfseries},
                        % %
                        % % the space between legend image and text:
                        % /tikz/every odd column/.append style={column sep=0cm},
                        % }
                    ]
                    \addplot[mark=*,red, line width=1.5pt]
                    coordinates{
                        (0,0.8403)
                        (10,0.6803)
                        (20,0.5653)
                        (50,0.5359)
                    }; 
                    \addlegendentry{SM}
                    \addplot[mark=square, blue, line width=1.5pt]
                    coordinates{
                        (0,0.8403)
                        (10,0.7020)
                        (20,0.5753)
                        (50,0.5653)
                    }; 
                    \addlegendentry{CM} 
                    \end{axis}
                \end{tikzpicture}
            }
            \caption[]%
            {{\small AUC-A}}
            % {{\small}}
            \label{fig:hit_g}
        \end{subfigure}
        \begin{subfigure}[b]{0.16\textwidth}   
            \centering 
            \resizebox{\linewidth}{!}{
                \begin{tikzpicture}
                    \begin{axis}[
                        scaled y ticks = false,
                        ymin=0.49, ymax=0.62,
                        ylabel={AUC},
                        xlabel={$\lambda$},
                        symbolic x coords={0,10,20, 50},
                        xtick=data,
                        scaled ticks=false, tick label style={/pgf/number format/fixed},
                        ymajorgrids=true,
                        % legend pos=south west,
                        legend style={at={(0.01,0.01)},
                        anchor=south west,
                        font=\fontsize{18}{5}\selectfont,
                        }
                        % legend style={
                        % draw=none, % ?
                        % text depth=0pt,
                        % at={(0.15,-0.30)},
                        % anchor=north west,
                        % legend columns=-1, 
                        % font=\fontsize{8}{5}\selectfont,
                        % % default spacing:
                        % column sep=0.1cm,
                        % % The text "Legend:"
                        % % /tikz/column 2/.style={column sep=0pt,font=\bfseries},
                        % %
                        % % the space between legend image and text:
                        % /tikz/every odd column/.append style={column sep=0cm},
                        % }
                    ]
                    \addplot[mark=*,red, line width=1.5pt]
                    coordinates{
                        (0,0.6064)
                        (10,0.5778)
                        (20,0.5432)
                        (50,0.5422)
                    }; 
                    \addlegendentry{SM}
                    \addplot[mark=square, blue, line width=1.5pt]
                    coordinates{
                        (0,0.6064)
                        (10,0.5449)
                        (20,0.5300)
                        (50,0.5243)
                    }; 
                    \addlegendentry{CM} 
                    \end{axis}
                \end{tikzpicture}
            }
            \caption[]%
            {{\small AUC-O}}
            % {{\small}}
            \label{fig:hit_g}
        \end{subfigure}
        \caption[]
        {Impact of the adversarial coefficient $\lambda$ on (a)-(c) recommendation performance $w.r.t$ Hit@5, and (d)-(f) classification quality of the attackers $w.r.t$ AUC. In the subfigure titles, G, A and O represent the sensitive features $gender$, $age$ and $occupation$, respectively. The reported results are from the Biased-MF model on the MovieLens dataset.} 
        \label{Ff1}
        \vspace{-10pt}
\end{figure*}

Furthermore, comparing the recommendation performance of SM and CM, we can see that SM performs better on most of the metrics on the two datasets, especially for the case of single feature. We analyze the phenomenon as follows: for CM method, although we aim at learning one filter function for each sensitive feature and wish that the filter function $f_k$ only filters out the information of the $k$-th sensitive feature, we actually train the combination of all the filter functions together. Such learning process will force the filter functions to also remove the information that they should not remove. For example, if one user chooses a mask $\textbf{Q}_u=\{gender,age\}$, and we use CM to generate user embedding as $\textbf{r}_u^*=\frac{1}{2}(f_g(\textbf{r}_u)+f_a(\textbf{r}_u))$. When we train the filter functions to require that discriminators cannot classify $gender$ and $age$ from $\textbf{r}_u^*$, $f_g$ may be forced to also remove some information of $age$ to satisfy the requirement, and it is the same for $f_a$. However, for SM method, we train one filter for each combination of sensitive features, thus the filter $f_g$ will only filter out the information of $gender$ as it has no contact with other features during the learning process. Therefore, comparing the recommendation performance of $f_g$ trained by SM and CM method, we will find that SM filter performs better as it keeps more information for making recommendations, and such performance is even more significant for the evaluation of the single feature case. 

However, SM method also has drawbacks as it will be infeasible when the potential number of combinations of the sensitive features is a lot. During the experiments, we found that SM usually needs more epochs to converge than CM, which will take more time for the training process. It is reasonable since SM method has more filter functions than CM, which needs extra epochs to make all filter functions be sampled and trained sufficiently. 

Therefore, we suggest using SM to achieve fair models when the number of sensitive features is small to keep better recommendation performance while use CM to handle the situations where there are too many combinations of sensitive features.

\subsection{Ablation Study}
We study the influence of the adversarial coefficient $\lambda$ on the recommendation performance and fairness in this section. As discussed before, $\lambda$ controls the trade-off between recommendation quality and fairness.
%Theoretically, the larger $\lambda$ is, the greater the influence of the discriminator loss will be in the whole loss, which means that we have a stricter demand for fairness, so that we must scarify more recommendation performance to meet the requirement. And when $\lambda\rightarrow \infty$, there exists a trivial solution that the filter functions always output a constant, which will result in a totally fair result but lose all the information for making accurate recommendations.
Theoretically, the larger $\lambda$ is, the greater the influence of the discriminator loss will be in the whole loss, which means that we have a stricter demand for fairness and may have to scarify more recommendation performance to meet the requirement. And when $\lambda\rightarrow \infty$, there is a trivial solution that the filter functions always output a constant, resulting in a fair result but losing all the information for making accurate recommendations. To verify the influence of $\lambda$, we draw the change of AUC scores and the recommendation performance with the change of $\lambda$ in Figure \ref{Ff1}. Since similar trend is observed for other recommendation models, datasets, metrics, and the combinations of sensitive features, we plot the results of Biased-MF on MoiveLens under three single feature cases to keep the figure clarity. We can see that the experimental results are consistent with our analysis above. i.e., with the increase of $\lambda$, the recommendation performance has been declining, while the AUC score is getting lower and lower, which means that the system is getting fairer.

\section{Conclusion}
\label{sec:conclusions}

In this paper, we study the fairness problem for users in recommender systems. To better assess fairness in recommendation, we adopt causal-based fairness notions to reason about the causal relations between the protected features and the predicted results, instead of merely considering the traditional association-based fairness notions.
We also consider personalized fairness which allows different users to have different fairness demands. 
Technically, to implement individual-level fairness for users, we approach counterfactual fairness for recommendation. 
% To Consider users' personalized fairness demands, we introduce a framework which allows to generate personalized fairness for users under the definition of counterfactual fairness, while can also cover non-personalized situation. 
We propose to generate feature-independent user embeddings to satisfy the counterfactual fairness requirements in recommendation, and we introduce an adversary learning method to learn such feature-independent user embeddings. Experiments on two real-world datasets with several shallow or deep recommendation algorithms show that our method is able to generate counterfactually fair recommendations for users with a desirable recommendation performance.

This work is one of our first steps towards personalized fairness under counterfactual notions in recommendation systems, and there is much room for future improvements. Except for the recommendation scenario that we considered in this work, we believe personalized fairness is also important for other intelligent systems such as search engines, social networks, language modeling and image processing, which we will consider in the future. 

\section*{Acknowledgement}
We thank the reviewers for the reviews and suggestions. This work was supported in part by NSF IIS-1910154 and IIS-2007907. Any opinions, findings, conclusions or recommendations expressed in this material are those of the authors and do not necessarily reflect those of the sponsors.

\bibliographystyle{ACM-Reference-Format}
\bibliography{sample-base}

\appendix









\end{document}